%% file: main.tex
\documentclass[11pt]{article}
\usepackage[margin=1in]{geometry}
\usepackage{libertine}
\usepackage{datetime}
\usepackage{hhline}
\usepackage{multirow}
\usepackage{stmaryrd}
\usepackage{booktabs}
\usepackage{enumitem}
\usepackage{wrapfig}
\usepackage{subcaption}
\usepackage[T1]{fontenc}
\usepackage[utf8]{inputenc}
\usepackage{pgfplots}
\usepackage{amsmath, amssymb}
\usepackage{amsthm}
\usepackage{color}
\usepackage{cases}
\usepackage{xparse}
\usepackage{xargs}
\usepackage{appendix}
\usepackage[boxed,commentsnumbered,noend]{algorithm2e}
\usepackage{algpseudocode}
\usepackage{verbatim, xspace}
\usepackage{tikz}
\usepackage{mathtools}
\usepackage{thm-restate}
\usepackage[framemethod=TikZ]{mdframed}

\usepackage{hyperref}
\usepackage{pgfplots}

\usepackage{xcolor}
\usepackage[noabbrev,capitalize]{cleveref}

\newenvironment{insight}
{\mdfsetup{%
    nobreak=true,
	middlelinecolor=gray,
	middlelinewidth=1pt,
	backgroundcolor=gray!10,
    innertopmargin=5pt,
	roundcorner=5pt}
\begin{mdframed}}
{\end{mdframed}}
\newcommand{\citet}[1]{\cite{#1}}
\usepackage[colorinlistoftodos,prependcaption,textsize=tiny]{todonotes}
\newcommandx{\unsure}[2][1=]{\todo[linecolor=green,backgroundcolor=green!25,bordercolor=green,#1]{\normalsize #2}}
\newcommandx{\improvement}[2][1=]{\todo[inline,linecolor=blue,backgroundcolor=blue!05,bordercolor=blue,#1]{\normalsize #2}}
\newcommandx{\info}[2][1=]{\todo[linecolor=yellow,backgroundcolor=yellow!25,bordercolor=yellow,#1]{#2}}
\newcommandx{\floatmodel}[2][1=]{\todo[inline,linecolor=red,backgroundcolor=yellow!25,bordercolor=yellow,#1]{#2}}
\newcommandx{\thiswillnotshow}[2][1=]{\todo[disable,#1]{#2}}
\newcommandx{\karol}[2][1=]{\todo[inline,linecolor=blue,backgroundcolor=blue!25,bordercolor=blue,caption={\normalsize \textbf{Karol}},#1]{\normalsize #2}}
\newcommandx{\antoine}[2][1=]{\todo[inline,linecolor=gray,backgroundcolor=red!25,bordercolor=red,caption={\normalsize
\textbf{Antoine}},#1]{\normalsize #2}}

\newtheorem{theorem}{Theorem}

\newtheorem{lemma}[theorem]{Lemma}

\newtheorem{claim}[theorem]{Claim}
\newtheorem{proposition}[theorem]{Proposition}
\newtheorem{observation}[theorem]{Observation}

\newtheorem*{maingoal*}{Main Question}

\numberwithin{theorem}{section}
\numberwithin{lemma}{section}
\numberwithin{claim}{section}
\numberwithin{corollary}{section}
\numberwithin{definition}{section}
\numberwithin{observation}{section}
\numberwithin{proposition}{section}

\newcommand{\Prb}{\mathbb{P}}

\newcommand{\norm}[1]{\left\lVert#1\right\rVert}
\newcommand{\norms}[1]{\lVert#1\rVert}

\newcommand{\floor}[1]{\left\lfloor #1 \right\rfloor}
\newcommand{\ceil}[1]{\left\lceil #1 \right\rceil}

\newcommand{\eps}{\varepsilon}

\newcommand{\Oh}{\mathcal{O}}
\newcommand{\Os}{\Oh^{\star}}

\newcommand{\N}{\mathbb{N}}
\newcommand{\Z}{\mathbb{Z}}
\newcommand{\Bb}{\mathcal{B}}

\newcommand{\Ll}{\mathcal{L}}

\newcommand{\Mm}{\mathcal{M}}

\newcommand{\real}{\mathbb{R}}

\newcommand{\poly}{\mathrm{poly}}

\newcommand{\Mod}[1]{\ (\mathrm{mod}\ #1)}

\newcommand{\prob}[2][]{\bold{Pr}_{#1}\left[ #2 \right]}

\newcommand{\angles}[1]{\langle {#1} \rangle}

\newcommand{\V}[1]{\normalfont{\mathrm{\textbf{#1}}}}
\newcommand{\trans}{\intercal}
\newcommand{\Vol}{\mathrm{Vol}}
\newcommand{\Ball}{\mathrm{Ball}}
\newcommand{\iv}[1]{\llbracket #1 \rrbracket}
\newcommand{\Rev}{\textbf{Rev}}
\newcommand{\dotprod}[1]{\langle #1 \rangle}
\newcommand{\CLLL}{\gamma_{\textsf{LLL}}}
\newcommand{\CGamma}{\Upsilon}
\newcommand{\LORange}{\Gamma_\textsf{LO}}
\newcommand{\Range}{\Gamma}

\renewcommand{\le}{\leqslant}
\renewcommand{\ge}{\geqslant}

\title{Improving Lagarias-Odlyzko Algorithm For Average-Case Subset Sum: Modular Arithmetic Approach}
\date{}

\author{
    Antoine Joux\footnote{CISPA Helmholtz Center for Information Security, Germany, \textsf{joux@cipa.de}. This work has been supported by the European Union's H2020 Programme under grant agreement number ERC-669891.}
    \and
    Karol W\k{e}grzycki\footnote{Saarland University and Max Planck Institute for Informatics,
        Saarbr\"ucken, Germany, \textsf{wegrzycki@cs.uni-saarland.de}. 
    This work is part of the project TIPEA that has
    received funding from the European Research Council (ERC) under the European Union's Horizon
    2020 research and innovation programme (grant agreement No 850979).}
}

\begin{document}

\maketitle

\thispagestyle{empty}
\input{chapters/abstract}

\begin{picture}(0,0)
\put(462,-170)
{\hbox{\includegraphics[width=40px]{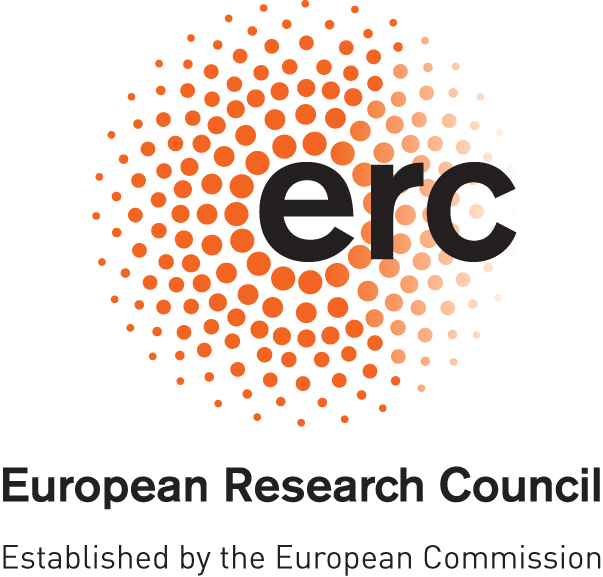}}}
\put(452,-230)
{\hbox{\includegraphics[width=60px]{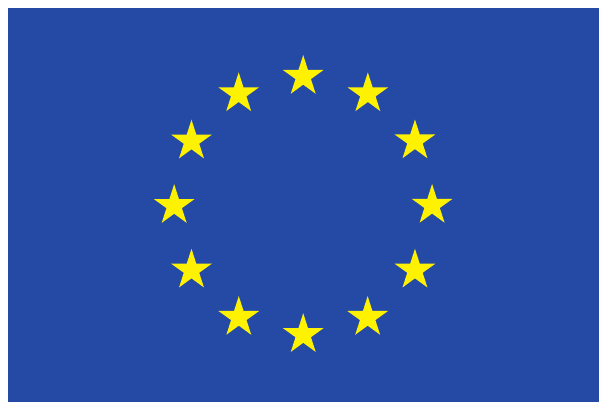}}}
\end{picture}

\clearpage
\setcounter{page}{1}

\input{chapters/introduction}

\input{chapters/prelim}

\input{chapters/lo}

\input{chapters/coppersmith}

\input{chapters/block-reduction}

\input{chapters/conclusion}

\bibliographystyle{abbrv}
\bibliography{bib}


\end{document}

%% file: chapters/abstract.tex
\begin{abstract} 

Lagarias and Odlyzko (J.~ACM~1985) proposed a polynomial time algorithm for
solving ``\emph{almost all}'' instances of the Subset Sum problem with $n$
integers of size $\Omega(\LORange)$, where $\log_2(\LORange) > n^2
\log_2(\gamma)$ and $\gamma$ is a parameter of the lattice basis reduction
($\gamma > \sqrt{4/3}$ for LLL).  The algorithm of Lagarias and Odlyzko is a
cornerstone result in cryptography. However, the theoretical guarantee on the
density of feasible instances has remained unimproved for almost 40 years.

In this paper, we propose an algorithm to solve ``\emph{almost all}'' instances
of Subset Sum with integers of size $\Omega(\sqrt{\LORange})$ after a single
call to the lattice reduction. Additionally, our argument allows us to solve the
Subset Sum problem for multiple targets while the previous approach could only
answer one target per call to lattice basis reduction. We introduce a modular
arithmetic approach to the Subset Sum problem. The idea is to use the lattice
reduction to solve a linear system modulo a suitably large prime.  We show that
density guarantees can be improved, by analysing the lengths of the LLL reduced
basis vectors, of both the primal and the dual lattices simultaneously. 

\end{abstract}

%% file: chapters/introduction.tex
\section{Introduction}

The Subset Sum problem, a cornerstone of computational complexity and
cryptography, has been the subject of extensive research due to numerous practical implications. In the Subset Sum
problem, we are given $n$ positive integers $a_1,\ldots,a_n \in \Z$ and a
target $T \in \Z$. The task is to decide if there exists a subset of the given
numbers that sums to $T$.

In particular, the study of Subset Sum has seen numerous applications in
cryptography~\cite{crypto-merkle78,dinur},
balancing~\cite{average-case-balancing}, and combinatorial
optimization~\cite{knapsack-book}. With these applications in mind, our goal is
to propose an efficient algorithm to solve the Subset Sum for a large range of
instances. Indeed, since the problem is NP-hard, it is very unlikely that we
could propose an efficient algorithm to solve all instances of the problem. The
natural next step is to characterize which instances of the problem can be
solved efficiently. On that front, there is a classical result of
Bellman~\cite{bellman} that proposes a polynomial time algorithm for Subset Sum
when the numbers are polynomially bounded, i.e., the instance is \emph{dense}.
This inspired a long line of pseudopolynomial-time algorithms (see
e.g.,~\cite{bringmann2017near,koiliaris2019faster,abboud2022seth,icalp21,pseudopolynomial-polyspace}).

On the other end of the spectrum, there is a result by Lagarias and
Odlyzko~\cite{LO} for the \emph{sparse} instances.  This work marked a
significant milestone by proposing a polynomial-time algorithm capable of
solving ``\emph{almost all}'' instances of Subset Sum when the size of the numbers is
$\Omega(\LORange)$, where $\log_2(\LORange) > n^2 \log_2{\gamma}$
and $\gamma$ is a parameter associated with the lattice reduction.  Lagarias and
Odlyzko~\cite{LO} used the LLL (Lenstra-Lenstra-Lov\'asz~\cite{lll}) algorithm as
a lattice basis reduction, for which $\gamma = \CLLL > \sqrt{4/3}$.

To this day, the algorithm of Lagarias and Odlyzko is a cornerstone of modern
cryptography and is featured in numerous text-books on the
field~\cite{crypto-book-1,crypto-book-2,nguyen2010lll}. This algorithm was
subsequently
implemented~\cite{radziszowski1988solving,schnorr1994lattice,brickell1984solving},
simplified~\cite{Frieze} and even improved assuming access to an exact
CVP~\cite{coster1992improved,DBLP:conf/fct/JouxS91}. Nevertheless, the algorithm
of Lagarias and Odlyzko is currently the only known method\footnote{There are
several extensions of the LLL algorithm tailored for solving Subset Sum in practice,
see~\cite{schnorr1994lattice,lamacchia1991basis}.} to
solve low-density instances of Subset Sum, both in theory and practice
(see~\cite{odlyzko1990rise} for survey).

In this paper, we revisit and improve the result of Lagarias and
Odlyzko~\cite{LO}. We show an algorithm solving almost all instances of Subset
Sum when the size of the numbers is $\Omega(\sqrt{\LORange})$. More formally, we prove the
following theorem.

\begin{theorem} \label{thm:main-theorem} \label{thm:main}
    Let $a_1,\ldots,a_n \in \Z_{\ge 0}$ be integers, selected uniformly at
    random from $\{1,\ldots,\ceil{\sqrt{\LORange}}\}$, where $\LORange =
    (\CLLL)^{n^2} \cdot 2^{\Theta(n\log{n})}$.
    With overwhelming probability, after a single call to a basis reduction algorithm with
    parameter $\CLLL$, we obtain a polynomial time tester, that for every given
    $T \in \Z$ decides whether
    \begin{displaymath}
        \text{there exist } e_1,\ldots,e_n \in \{0,1\} \text{ such that }
        \sum_{i=1}^n e_i \cdot a_i = T.
    \end{displaymath}
    Our tester and its creation method are both deterministic.
    Their success probability only depends on the random choice of the integers 
    $a_1,\ldots,a_n$ and is at least $1-2^{-\Omega(n \log{n})}$.
\end{theorem}

The improvement offered by~\cref{thm:main-theorem} is two-fold. First, in terms
of admissible range, our algorithm provides improvement from $\Omega(\LORange)$
down to $\Omega(\sqrt{\LORange})$ on numbers in the instance. Second, our algorithm is
\emph{oblivious} to the target. Namely, after a single call to the lattice basis
reduction, with input numbers $a_1,\ldots,a_n$ we can deterministically decide
on an answer in polynomial time. In that sense, our argument provides a succinct
certificate for ``almost all'' instances of sparse Subset Sum. Precisely this
question was already considered by Furst and Kannan~\cite{furst1989succinct} who
used lattice basis reduction to bound the proof complexity of Subset Sum.  They
strengthened the Lagarias and Odlyzko algorithm and showed that ``almost all''
instances with numbers in $\{1,\ldots,\LORange\}$ have a polynomial size proof.
Our argument again reduces the range down to $\{1,\ldots,\ceil{\sqrt{\LORange}}\}$, and our certificates are
simpler. In fact, the data-structure behind~\cref{thm:main-theorem}
invokes a single matrix-vector multiplication.

\subsection{Our Technique: Modular Arithmetic Approach}

Now, we elaborate on the techniques that we introduce. In this note, we let
$\iv{a}_p$ be a unique integer in $(-p/2,p/2)$ such
that $a \equiv \iv{a}_p \Mod p$ (see~\cref{sec:prelim} for notation).
Consider vectors $\V{a} = (a_1,\ldots,a_n) \in \Z^n$, $\V{e} =
(e_1,\ldots,e_n) \in \{0,1\}^n$ and an integer $T$ such that
\begin{align*}
    e_1 \cdot a_1 + e_2 \cdot a_2 + \ldots + e_n \cdot a_n & = T.\\
\intertext{Obviously, for any integers $\mu,p \in \Z$ the following also holds}
    e_1 \cdot \iv{\mu a_1}_p + e_2\cdot \iv{\mu a_2}_p +  \ldots + e_n \cdot \iv{\mu
    a_n}_p & \equiv \iv{\mu T}_p \Mod p.
\end{align*}
Next, assume, that we have managed to find $\mu$ and $p$, such that:
\begin{align*}\label{eq:assumption}
    \sum_{i=1}^n |\iv{\mu a_i}_p| < p/2.
\end{align*}
Then there are no overflows modulo $p$ and actually, we can write
\begin{displaymath}
    e_1 \cdot \iv{\mu a_1}_p + e_2\cdot \iv{\mu a_2}_p + \ldots + e_n \cdot \iv{\mu a_n}_p = \iv{\mu T}_p
    .
\end{displaymath}
We will show that if the density is small, then there are integers $\mu_1,\ldots,\mu_n \in
\Z$ such that for each of them \emph{(a)} the $\ell_1$ norm is bounded, i.e., $\norm{\iv{\mu_i \V{a}}_p}_1 < p/2$, and
\emph{(b)} the vectors $\iv{\mu_i \V{a}}_p$ are linearly independent. More formally we
will prove the following theorem.

\begin{restatable}[Modular Arithmetic Approach]{theorem}{multipliers}\label{thm:multipliers}
	Let $\LORange \coloneqq \Theta((\CLLL)^{n^2} \cdot 2^{4 n \log_2{n}})$ and let $p = \Theta((\CLLL)^{0.5 n^2})$ be a prime.
There exists a polynomial time algorithm, that 
    given integers $a_1,\ldots,a_n \in \Z$ selected uniformly at random from
	$\{1,\ldots,\ceil{\sqrt{\LORange}}\}$ outputs, with overwhelming probability, integers
    $\mu_1,\ldots,\mu_n \in \Z$ such that the following $n
    \times n$ matrix
    \begin{displaymath}
        \Mm_p  \coloneqq
        \begin{pmatrix}
            \iv{\mu_1 a_1}_p & \iv{\mu_1 a_2}_p & \cdots & \iv{\mu_1 a_n}_p\\
            \iv{\mu_2 a_1}_p & \iv{\mu_2 a_2}_p & \cdots & \iv{\mu_2 a_n}_p\\
            \vdots & & & \vdots \\
            \iv{\mu_n a_1}_p & \iv{\mu_n a_2}_p & \cdots & \iv{\mu_n a_n}_p\\
        \end{pmatrix}
    \end{displaymath}
    has full rank. Moreover, for every $1 \le i \le n$ it holds that:
    \begin{displaymath}
        \norm{(\iv{\mu_i a_1}_p,\ldots,\iv{\mu_i a_n}_p)}_1 < \frac{p}{2}
        .
    \end{displaymath}
    More precisely, the procedure succeeds with probability $\ge
    1-2^{-\Omega(n \log{n})}$, which only depends on the
    random choice of integers $a_1,\ldots,a_n$.
\end{restatable}

Observe, that the algorithm for~\cref{thm:main} follows easily with the matrix
$\Mm_p$ in hand. Namely, consider $\Mm_p \cdot \V{e}^\trans = \V{t}$, where $\V{t}
\coloneqq (\iv{\mu_1 T}_p,\ldots,\iv{\mu_n T}_p)^\trans$. First, because the
$\ell_1$ norm of each row of $\Mm_p$ is bounded by $p/2$, vector $\V{e}$ is also
a solution to $\dotprod{\V{a},\V{e}} = T$.  Second, by the fact that $\Mm_p$ has
full-rank, such a vector $\V{e}$ is \emph{unique}.  Hence, it is enough to
compute the inverse of $\Mm_p$ (which is possible because $\Mm_p$ has full-rank)
and verify that the vector $(\Mm_p)^{-1} \cdot \V{t}$ is in $\{0,1\}^n$
(see~\cref{sec:Coppersmith} for the full proof). 

Now, we briefly sketch the proof of~\cref{thm:multipliers}.  Our algorithm is
inspired by Howgrave-Graham's revisit~\cite{DBLP:conf/ima/Howgrave-Graham97} of
Coppersmith's method for finding short roots of univariate modular
equations~\cite{DBLP:conf/eurocrypt/Coppersmith96}. In a nutshell, we construct
a lattice (very similar to the original lattice of Lagarias and Odlyzko) for
Subset Sum modulo a prime number $p$. Analogously to the original approach of
Lagarias and Odlyzko we bound the length of vectors in this lattice. The caveat
is that we can also bound the lengths of the vectors in the \emph{dual} lattice.
This crucially allows us to double the density of the admissible instances.

There are two main differences between our proposal and
Howgrave-Graham's~algorithm~\cite{DBLP:conf/ima/Howgrave-Graham97} (and the
original~Coppersmith's method~\cite{DBLP:conf/eurocrypt/Coppersmith96}) Firstly,
the number of variables to consider in our case is much larger.  Secondly, we
only need to construct linear polynomials in these variables rather than
considering large degree polynomials. 

\subsection{Extension to Block Basis Reduction}

We would like to comment on the possible improvements of using stronger lattice
reduction.
 Specifically, the only aspect of the Lagarias-Odlyzko algorithm
that relies on the LLL algorithm is the fact that it provides a
$(\CLLL)^n$-approximate solution to the Shortest Vector Problem with
$\CLLL>\sqrt{4/3}$.  Instead of the LLL algorithm, it is natural to apply
block-basis reduction (e.g.,~\cite{schnorr}) with parameter $\CGamma$ arbitrarily
close to $1$. This way, one expects a better approximation
of SVP, consequently improving the range $\Range$. The obvious drawback of this
is the increased running time. 

Our approach, in principle, should allow us to reduce the bitlengths of the Subset Sum instances that
can be solved regardless of the concrete algorithm behind the lattice reduction.
Unfortunately, adapting our approach to work with every basis lattice reduction
requires a separate analysis which is much more technical than for LLL.
Nevertheless, the takeaway message is as follows:

\begin{insight}
    The Modular Arithmetic Approach allows us to solve
    almost all instances of Subset Sum with integers $\Omega(\sqrt{\Range})$ with
    a single call to lattice reduction, whereas the previous approach guaranteed to
    solve only instances with integers $\Omega(\Range)$ while using the same lattice
    reduction algorithm.
\end{insight}

Since the LLL algorithm is the most popular we decided first to present our
algorithms using this algorithm. In~\cref{sec:block}, we demonstrate that our
approach works for the textbook block-reduction presented by Gama and
Nguyen~\cite{gama2008finding,nguyen2010lll}.
\begin{theorem}~\label{thm:block}
    Let $\CGamma > 1$ be a parameter of lattice reduction
    in~\cite{gama2008finding}.
    Let $a_1,\ldots,a_n \in \Z$ be positive integers, selected uniformly at
    random from range $\{1,\ldots,\ceil{\sqrt{\Range}}\}$, where $\Range =
    \CGamma^{n^2} \cdot 2^{\Theta(n\log{n})}$.
    With overwhelming probability, after a single call to a basis reduction algorithm
    in~\cite{gama2008finding} with parameter $\CGamma$, we obtain a polynomial time tester, that for every given
    $T \in \Z$ decides whether
    \begin{displaymath}
        \text{there exist } e_1,\ldots,e_n \in \{0,1\} \text{ such that }
        \sum_{i=1}^n e_i \cdot a_i = T,
    \end{displaymath}
    Our tester and its creation method are both deterministic.
    Their success probability only depends on the random choice of the integers 
    $a_1,\ldots,a_n$ and is at least $1-2^{-\Omega(n \log{n})}$.
\end{theorem}

In comparison, the standard analysis of Lagarias-Odlyzko algorithm requires the
range of integers to be of order $\Range = \CGamma^{n^2} \cdot 2^{\Theta(n\log{n})}$ when the lattice reduction
with parameter $\CGamma > 1$ is used.

\subsection{Related Work}

Subset Sum is an NP-hard problem, closely related to knapsack. The currently
fastest algorithm to solve the problem exactly in the worst-case settings runs
in $\Os(2^{n/2})$ time~\cite{horowitz-sahni}.  Schroeppel and
Shamir~\cite{schroeppel} show that Subset Sum admits a time-space tradeoff,
i.e., an algorithm using $\mathcal{S}$ space and $2^{n}/\mathcal{S}^2$ time for
any $\mathcal{S} \le \Os(2^{n/4})$. This tradeoff was improved by
\cite{subsetsum-tradeoff} for almost all tradeoff parameters (see also
\cite{dinur,NederlofW21}).  In the average-case setting, when $d = \Theta(1)$,
Howgrave{-}Graham and Joux~\cite{generic-knapsack1} showed  $\Os(2^{0.337n})$
time and $\Os(2^{0.256n})$ space algorithm for Subset Sum. This was subsequently
improved by Becker, Coron and Joux~\cite{generic-knapsack2} to $\Os(2^{0.291n})$
time and space.

The exact running time of the Subset Sum was also analysed in the parameterized
setting~\cite{stacs2015,stacs2016}, quantum
regime~\cite{DBLP:conf/esa/AllcockHJKS22,helm2018subset}, the polynomial space
regime~\cite{polyspace-stoc2017}. From the perspective of pseudopolynomial
algorithms, Subset Sum has also been the subject of recent stimulative
research~\cite{bringmann2017near,koiliaris2019faster,abboud2022seth,icalp21,pseudopolynomial-polyspace}.

%% file: chapters/prelim.tex
\section{Preliminaries}\label{sec:prelim}

Through the paper, we let $[n] \coloneqq \{1,\ldots,n\}$.  We use $\equiv_p$ to
denote equivalence modulo $p$, i.e., $a \equiv_p b$ iff there exists $k \in
\mathbb{Z}$ such that $a = k\cdot p + b$.  For integers $a,b$ we let $\iv{a}_b$
be the unique value in set $\{-\ceil{b/2}+1,\ldots,\floor{b/2}\}$ such that
$\iv{a}_b \equiv_b a$. $\Os(\cdot)$ notation hides factors polynomial in $n$.

We consider $\real^n$ with the usual Euclidean topology. We use bold letters to
denote vectors. The inner product of two vectors $\V{x} = (x_i)_{i=1}^n$ and
$\V{y} = (y_i)_{i=1}^n$ is denoted by $\dotprod{\V{x},\V{y}} = \sum_{i=1}^n
x_i \cdot y_i$. We denote the $\ell_p$ norm as $\norm{\V{x}}_p \coloneqq
(\sum_{i=1}^n (x_i)^p)^{1/p}$. Often, we omit the subscript $p$ which means that
we consider the Euclidean case with $p=2$.  We denote the $d$-dimensional Euclidean ball of
radius $R$ centered at~$\V{0}$ by $\Ball_d(R)$. We use the Stirling formula to approximate the
volume of the $d$-dimensional ball, namely:
\begin{displaymath}
    \Vol(\Ball_d(R)) \sim \frac{1}{\sqrt{d\pi}} \left(\frac{2\pi
    e}{d}\right)^{d/2} \cdot R^d = 2^{-\Theta(d \log(d))} \cdot R^d.
\end{displaymath}

\subsection{Background on Lattices}

A \emph{lattice} in $\real^n$ is a discrete subgroup of
$(\real^n,+)$. This implies that a lattice is a non-empty set $L \subseteq
\real^n$ such that for all $\V{x},\V{y} \in L$ it holds that $\V{x}-\V{y} \in L$.
Any subgroup of $(\Z^n,+)$ is a lattice, and it is called an 
\emph{integral lattice}.

Let $\V{b}_1,\ldots,\V{b}_m$ be vectors in $\real^n$. Let
$\Ll(\V{b}_1,\ldots,\V{b}_m)$ be the set of integral combinations of these
vectors, namely:
\begin{displaymath}
    \Ll(\V{b}_1,\ldots,\V{b}_m) \coloneqq \left\{ \sum_{i=1}^m \ell_i \cdot \V{b}_i
    \text{ such that } \ell_1,\ldots,\ell_m \in \Z \right\}.
\end{displaymath}
Then, $\Ll(\V{b}_1,\ldots,\V{b}_m)$ is a lattice, specifically the lattice spanned
by $\V{b}_1,\ldots,\V{b}_m.$ Furthermore, when the vectors  $\V{b}_1,\ldots,\V{b}_m$
are linearly independent (over $\real$), we say that they form a basis
of $\Ll(\V{b}_1,\ldots,\V{b}_m).$ 
In this case, every vector $\V{v} \in L$ uniquely
decomposes as an integral combination of the basis vectors , i.e., for every $\V{v} \in L$
there exists unique $\alpha_1,\ldots,\alpha_m \in \Z$ such that $\V{v} =
\sum_{i=1}^m \alpha_i \cdot \V{b}_i$~\cite[Definition 5,
p. 26]{nguyen2010lll}. It is well-known that every lattice admits
basis. The cardinality of any basis is an invariant of the lattice and
is called its \emph{dimension} (or sometimes referred to as its rank). This
quantity, denoted by $\mathrm{dim}(L)$ is equal to the
dimension of its linear span (denoted $\mathrm{span}(L)$) as a vector
space over $\real.$

We say that a lattice $L \subseteq \real^n$ has full rank when
$n = \mathrm{dim}(L)$. Since any lattice of $\real^n$ admits some
basis, we assume w.l.o.g. that every lattice we consider is given in
the form $\Ll(\V{b}_1,\ldots,\V{b}_m)$ for some $\V{b}_i$'s in
$\real^n$ where $m = \mathrm{dim}(L)$.

Thus, $\V{b}_1,\ldots,\V{b}_m \in \real^n$ are linearly independent vectors. Their
Gram-Schmidt orthogonalization (GSO) is the orthogonal family
$(\V{b}^\ast_1,\ldots,\V{b}^\ast_m)$ defined inductively as follows: 

\begin{align*}
 & \V{b}_1^\ast = \V{b}_1 \quad \text{and}\\
 &  \V{b}_i^\ast = \V{b}_i - \sum_{j=1}^{i-1} \mu_{i,j} \V{b}_j^\ast \; \text{ where
    } \;
    \mu_{i,j} \coloneqq \frac{\dotprod{\V{b}_i , \V{b}_j^\ast}}{\norm{\V{b}_j^\ast}^2} \;\text{
    for all }\; 1 \le j < i \le m.
\end{align*}

We can define the volume of the lattice $L$ spanned by $\V{b}_1,\ldots,\V{b}_m$ as
\begin{displaymath}
    \Vol(L) = \prod_{i=1}^m \norm{\V{b}_i^\ast},
\end{displaymath}
and remark that it is a lattice invariant and does not depend on the
choice of basis. 

Other classical invariants of a lattice are its successive minima. By
definition, the $i$th successive
minimum, denoted as $\lambda_i(L)$ is the smallest real number for which there exist $i$
linearly independent vectors in $L$ each of length at most $\lambda_i(L)$. In
particular, $\lambda_1(L)$ is the length of the shortest nonzero vector in $L$. Note
that $\lambda_1(L) \le \lambda_2(L) \le \ldots \le \lambda_m(L)$.

\subsection{Dual Lattice}

For a lattice $L \subset \real^n$ the \emph{dual} lattice of $L$ is defined as:
\begin{displaymath}
    L^\dagger \coloneqq \{ \V{y} \in \mathrm{span}(L) \text{ such that }
    \angles{\V{x},\V{y}} \in \Z \text{ for all } \V{x} \in L \}
    .
  \end{displaymath}

When $L$ is an $m$-dimensional lattice of $\mathbb{R}^n$ then $L^\dagger$ is
also an $m$-dimensional lattice of $\mathbb{R}^n$.
\iffalse
More precisely, if $L$ is
spanned by the columns of the basis matrix $B$, then $L^\dagger$ is spanned by the
columns of:
\[B^\dagger = B \left(B^\trans B\right)^{-1}
\]
Furthermore, it holds that $\Vol(L) \cdot \Vol(L^\dagger) = 1$. Finally, if we
reverse the order of the columns of $B^\dagger$, we obtain the so-called dual
basis $D^\dagger$ of $B$. The GSO of dual basis is related by the following
relations:
\begin{displaymath}
    \norm{\V{b}_i^\ast} \cdot \norm{\V{d}_{m-i+1}^\ast} = 1 \quad
    \mbox{for every}\  i \in [m].
  \end{displaymath}

  In order to make explicit, that the order of the basis vectors is reversed, we
write:
\[
    D^\dagger = \Rev\left(B \left(B^\trans B\right)^{-1}\right).
\]

\else
More precisely, if $L$ is
spanned by the rows of the basis matrix $B$, then $L^\dagger$ is spanned by the
rows of:
\[B^\dagger = \left(B B^\trans\right)^{-1}B
\]
Furthermore, it holds that $\Vol(L) \cdot \Vol(L^\dagger) = 1$. Finally, if we
reverse the order of the rows of $B^\dagger$, we obtain the so-called dual
basis $D^\dagger$ of $B$. The GSO of dual basis is related to the GSO
of the primal basis by the following relations:
\begin{displaymath}
    \norm{\V{b}_i^\ast} \cdot \norm{\V{d}_{m-i+1}^\ast} = 1 \quad
    \mbox{for every}\  i \in [m].
  \end{displaymath}
In order to make explicit that the order of the (row) basis vectors is reversed, we
write:
\[
    D^\dagger = \Rev\left(\left(B B^\trans\right)^{-1} B\right).
\]

\fi

\subsection{The LLL algorithm and its property}
Lenstra, Lenstra and Lov{\'a}sz~\cite{lll} presented a polynomial time procedure
that, given an arbitrary basis of a lattice computes a so-called
reduced basis with better properties. Extensions of LLL can also
compute a reduced basis from an arbitrary generating family.

The main parameter of LLL is a real $\delta \in (0.25,1)$.  A secondary parameter $\mu\in [0.5,1)$ is also useful
when considering implementations of LLL with finite precision on
orthogonalised vectors. The initial LLL with exact rational arithmetic~\cite{lll}
uses $\mu=1/2$. Let $\V{b}_1,\ldots,\V{b}_m \in \real^n$ be basis of lattice $L$ and 
$\V{b}_1^\ast,\ldots,\V{b}^\ast_m \in \real^n$ be its Gram-Schmidt orthogonalized basis of
$L$. We say that $(\V{b}_1,\ldots,\V{b}_m)$ is \emph{$(\delta,\mu)$-LLL
  reduced} if and only if the following conditions are satisfied:
\begin{align*}
    \text{Size Reduction: }\quad & |\mu_{i,k}| \le \mu & \text{ for every } 1 \le k < i \le m\\
    \text{Lov\'asz Condition: } \quad &
\delta \norm{\V{b}_i^\ast}^2  \le \norm{\V{b}_{i+1}^\ast}^2 + \mu_{i+1,i}^2
    \norm{\V{b}_i^\ast}^2 & \text{ for every } 1 \le i < m
\end{align*}

Recall that $\mu_{i,k} = \frac{\dotprod{\V{b}_i ,
\V{b}_k^\ast}}{\dotprod{\V{b}_k^\ast , \V{b}_k^\ast}}$. Both of these conditions
yield an \emph{implied parameter} $\CLLL \coloneqq \frac{1}{\sqrt{\delta -
\mu^2}}$, such that:

\begin{align}
    \norm{\V{b}_i^\ast} \le \CLLL \cdot \norm{\V{b}_{i+1}^\ast} \text{ for
    every } 1 \le i < m \label{ineq:LLL3}
\end{align}

\cref{ineq:LLL3} is the strongest for $\CLLL = \sqrt{4/3}+\eps$ for some $\eps
> 0$. This is achieved by taking $\mu$ sufficiently close to $1/2$ and $\delta$
sufficiently close to $1$. In the sequel, we focus on the parameter $\CLLL$
and express our results in terms of it. Note, that the same inequality holds
for the dual lattice.

The running time of LLL algorithm~\cite{lll} is $\Oh(m^5 n \cdot 
\text{polylog}(B))$, where $B$ is the Euclidean length of the longest
basis)~\cite{lll,nguyen2010lll} (see~\cite{10.1007/978-3-031-38548-3_1} for
recent development).

%% file: chapters/lo.tex
\section{State-of-the-art: Lagarias-Odlyzko Algorithm}

Assume that $\V{a} = (a_1,\ldots,a_n)$ is a vector in $\mathbb{Z}^n$ where each $a_i$ is chosen independently uniformly at
random from $\{1,\ldots,\Range\}$ and $\V{e} = (e_1,\ldots,e_n) \in \{0,1\}^n$ is
chosen independently at random. Set $T = \dotprod{\V{a},\V{e}}$. Then, clearly
$\V{e}$ is a solution to the equation:
\begin{displaymath}
    \sum_{i=1}^n a_i x_i = T \text{ such that } x_i \in \{0,1\} \text { for every } i \in \{1,\ldots,n\}.
\end{displaymath}
Here, we recall the argument of Lagarias and Odlyzko who show that if $\Range \ge
\Omega((\CLLL)^{n^2 + o(n^2)})$ and $\CLLL > \sqrt{4/3}$, then the LLL
algorithm is sufficient to solve an instance $a_1,\ldots,a_n$ with target $T$ of
Subset Sum in polynomial time and with probability $\ge 1-2^{-\Omega(n\log{n})}$.

Note that $T \le \sum_{i=1}^n a_i$, and we can assume that
$\norm{\V{e}}_1 \le \frac{n}{2}$ as otherwise, we can 
consider the same instance with target $\left(\sum_{i=1}^n a_i -
T\right)$. Let $K \coloneqq \ceil{n(\CLLL)^n}$ be a sufficiently large integer. Lagarias
and Odlyzko~\cite{LO} construct the following row-generated lattice:

\begin{displaymath}
L \coloneqq 
\begin{pmatrix}
    1 & 0 & \cdots & 0& a_1 \cdot K \\
    0 & 1 & \cdots & 0& a_2  \cdot K \\
    \vdots & \vdots & \ddots & \vdots & \vdots \\
    0 & 0 & \cdots & 1 & a_n \cdot K \\
    0 & 0 & \cdots & 0 & T \cdot K \\
\end{pmatrix}.
\end{displaymath}
Let $\bar{\V{e}} = (e_1,\ldots,e_n,0)$ such that $\sum_{i=1}^n e_i a_i = T$.
Observe that $\bar{\V{e}} \in L$ and $\norm{\bar{\V{e}}} \le \sqrt{n/2}$.
Therefore, the vector $\bar{\V{e}}$ is a \emph{short vector} of $L$.

The idea of Lagarias and Odlyzko is to use the algorithm of Lenstra, Lenstra
and Lov\'asz~\cite{lll}. Let $\V{v}$ be a minimum length non-zero vector in
$L$. The LLL algorithm guarantees finding a nonzero vector $\V{x} \in L$ with:
\begin{displaymath}
    \norm{\V{x}} \le (\CLLL)^{n} \norm{\V{v}}.
\end{displaymath}
Here, $\CLLL$ is an implied parameter of the LLL algorithm, and we can guarantee
that when $\CLLL > \sqrt{4/3}$, the LLL algorithm runs in polynomial time
(see~\cref{sec:prelim}). Note that the vector $\V{x}$ found by the LLL satisfies
\begin{displaymath}
    \norm{\V{x}} \le (\CLLL)^n \norm{\bar{\V{e}}} \le (\CLLL)^n\sqrt{n/2} < K
\end{displaymath}
The strategy is, hence, to run the LLL algorithm and hope that the vector $\V{x}$ returned by
LLL is either $\bar{\V{e}}$ or $-\bar{\V{e}}$. There is of course the possibility
that there are other spurious vectors of length less than $K$. The analysis by
Lagarias and Odlyzko demonstrates that this is, however, highly unlikely.
More precisely, they show the following statement:
\begin{claim}[cf.,\cite{Frieze}] \label{claim:LO}
    $
    \prob{\text{exists } \V{x} \in L \setminus \{ k \cdot \V{e} \mid k \in
        \mathbb{Z} \} \text{ such
            that } \norm{\V{x}} < K
        } \le \frac{(2 K+1)^{n+1}}{\Range}.
    $
\end{claim}
Note that this is $2^{-\Omega(n\log{n})}$ when $\Range \ge \Omega((\CLLL)^{n^2
+o(n^2)})$. Therefore, by considering the opposite event, we can guarantee
that the vector $\bar{\V{e}}$ is found with a probability at least $1 - 2^{-\Omega(n \log{n})}$.  We
include the proof of~\cref{claim:LO} due to~\cite{Frieze} as it serves as an
introduction to our technique.
\begin{proof}[Proof of~\cref{claim:LO}]
    Let $\V{w} = (w_1,\ldots,w_{n+1}) \in L$. Observe that if $w_{n+1} \neq 0$
    then already $\norm{\V{w}} \ge K$. Hence, it must hold that $w_{n+1} = 0$. Let
    \begin{displaymath}
        \mathcal{W} \coloneqq \{ \V{w} \in L \text{ such that } \norm{\V{w}} <
            K \text{ and } w_{n+1} = 0 \text{ and } \V{w} \neq k \cdot
        \bar{\V{e}} \text{ for any } k \in \mathbb{Z} \}.
    \end{displaymath}
    Now, it suffices to bound the probability that $\mathcal{W}$ is empty.
    If, however, $\mathcal{W}$ is nonempty, there exists $\V{w} =
    (w_1,\ldots,w_n,0) \in \mathbb{Z}^{n+1}$
    and $\ell \in \mathbb{Z}$ that satisfy:
    \begin{align}
        \label{eq:A0empty}
        \sum_{i=1}^n a_i \cdot w_i = \ell \cdot T.
    \end{align}
    Let us now fix this vector $\V{w}$ and integer $\ell$. Let $\V{z} = (z_1,\ldots,z_n) \in \mathbb{Z}^n$ be such
    that $z_i = w_i - e_i \cdot \ell$.
    Observe that $\dotprod{\V{a},\V{z}} = 0$. Without loss of generality, we
    can assume that $z_1 \neq 0$ and we let $Z \coloneqq -(\sum_{i=2}^n
    a_i z_i/z_1)$. Hence,
    \begin{displaymath}
        \prob{\dotprod{\V{a},\V{z}}=0} = \prob{a_1 = Z} = \sum_{i=1}^\Range
        \prob{a_1 = i \;|\; Z =i} \cdot \prob{Z = i}.
    \end{displaymath}
    As $a_1$ and $Z$ are independent, this is bounded by
    \begin{displaymath}
        \sum_{i=1}^\Range \frac{1}{\Range} \cdot \prob{Z = i} \le
        \frac{1}{\Range}.
    \end{displaymath}
    Therefore, a fixed $\V{w}$ and $\ell$ satisfy Equation~\eqref{eq:A0empty}
    with probability at most $1/\Range$. Note that the number of $\V{w} \in
    \mathbb{Z}^{n+1}$ and $\ell \in \mathbb{Z}$ such that $\norm{\V{w}} < K$,
    $w_{n+1} = 0$ and $|\ell| \le 2K$ is
    at most $(2K+1)^{n+1}$ and the proof concludes.
\end{proof}

Note, that the proof of~\cref{claim:LO} uses lattice reduction as a blackbox.
If one were to use lattice reduction with parameter $\CGamma > 1$, then the
range guaranteed by~\cref{claim:LO} would be $\Range = \Theta((\CGamma)^{n^2 +
o(n^2)})$. In the next sections, we will show a method to improve the admissible
range down to $\Range = \Theta((\CGamma)^{0.5 n^2 + o(n^2)})$ that uses more properties of lattice
reduction. In~\cref{sec:Coppersmith} we analyse it with the standard LLL. Then
in section~\cref{sec:block} we analyse it with a textbook block lattice
reduction of Gama and Nguyen~\cite{gama2008finding}.

%% file: chapters/coppersmith.tex
\section{Modular arithmetic approach}\label{sec:Coppersmith}

In this section, we focus on the proof of the following statement.

\multipliers

Before we do so, let us show how to use~\cref{thm:multipliers} to solve the Subset
Sum formally.

\begin{proof}[Proof of Theorem~\ref{thm:main-theorem} assuming~\cref{thm:multipliers}]
    We start by describing the algorithm. First, we use~\cref{thm:multipliers}
    to compute the integers $\mu_1,\ldots,\mu_n,p$. This gives us the matrix
    $\Mm_p$. We know that this matrix has full rank so we
    can compute its inverse with the Gaussian elimination algorithm. This concludes the
    description of the preprocessing phase. 

    Now, given a target $T$ we compute the vector $\V{t}_p \coloneqq (\iv{\mu_1
    T}_p,\ldots,\iv{\mu_n T}_p)^\trans$ and compute a candidate solution
	$\V{e}^\trans = \Mm_p^{-1} \cdot \V{t}_p$. Finally, we check that $\V{e}$ is indeed correct. Namely, if $\sum_{i=1}^n e_{i} a_j = T$ and $\V{e} \in \{0,1\}^n$ we return
    $\V{e}$ and $\bot$ otherwise. This concludes the description of the algorithm.
    Clearly, the algorithm runs in polynomial time. Moreover, a query takes only
    $\Oh(n)$ arithmetic operations (on numbers from the range $[R]$). The
    success probability of the algorithm comes exclusively from a single
    application of~\cref{thm:multipliers}. Therefore, it remains to argue about
    the correctness.

    Let $e_1,\ldots,e_n \in \{0,1\}^n$ be a solution to  $\sum_{i=1}^n
    a_i e_i = T$ and assume that our algorithm returns $\bot$.
    Because the norm $\norm{\iv{\mu_i \V{a}}_p}_1$ is
    bounded by $p/2$ it holds that:
    
\begin{displaymath}
\begin{pmatrix}
    \iv{\mu_1 a_1}_p & \iv{\mu_1 a_2}_p & \cdots & \iv{\mu_1 a_n}_p\\
    \iv{\mu_2 a_1}_p & \iv{\mu_2 a_2}_p & \cdots & \iv{\mu_2 a_n}_p\\
    \vdots & & & \vdots \\
    \iv{\mu_n a_1}_p & \iv{\mu_n a_2}_p & \cdots & \iv{\mu_n a_n}_p\\
\end{pmatrix} 
\begin{pmatrix}
    e_1 \\
    e_2 \\
    \vdots\\
    e_n
\end{pmatrix} 
= 
\begin{pmatrix}
    \iv{\mu_1 T}_p \\
    \iv{\mu_2 T}_p \\
    \vdots \\
    \iv{\mu_n T}_p \\
\end{pmatrix} 
\end{displaymath}
    Since the matrix $\Mm_p$ has full rank, there exists a unique solution to the above linear system. Hence
    $(e_1,\ldots,e_n)^\trans = \Mm_p^{-1} \cdot \V{t}_p$. Finally,
    note that we only return this solution if all its coordinates are in
    $\{0,1\}^n$, thus ensuring an incorrect solution is never returned.
\end{proof}

From now on we will focus on the proof of~\cref{thm:multipliers}.

\subsection{Generating Family}

Let $p \in \mathbb{N}$ be the prime fixed in~\cref{thm:multipliers}.  In
particular, $p$ is odd. Because the numbers $a_1,\ldots,a_n$ are generated at
random in the interval $\{1,\ldots,\sqrt{\LORange}\}$ which has length much
larger than $p$, the modular reductions $a_i \pmod{p}$ are very close to
uniform modulo $p$. For simplicity of the analysis we want to assume that they
are uniform. This can be achieved by using the rejection sampling to discard any event
when at least one $a_i$ is greater than $\sigma_p$, where $\sigma_p$ is the
largest multiple of $p$ such that $\sigma_p < \sqrt{\LORange}$. This is valid,
because the probability that specific $a_i$ is discarded is:

\begin{displaymath}
	\Prb\left[a_i > \sigma_p\right] \le \frac{|\sqrt{\LORange} - \sigma_p|}{\sqrt{\LORange}} \le 2^{-\Omega(n \log{n})}.
\end{displaymath}

Therefore, by the union bound, probability that we do not discard any of
$a_1,\ldots,a_n$ is at least $1-2^{-\Omega(n \log{n})}$. Similarly, we can
assume that $a_1$ is not a multiple of $p$. Hence, there exists an integer
$a_1^{-1}$ such that $a_1 \cdot a_1^{-1} \equiv_p 1$. 

Consider the lattice generated by the rows of the following matrix:
\begin{equation}
     \Ll =
 \begin{pmatrix}
     a_1 & a_2 & a_3 & \cdots & a_n \\
     p & 0 & 0 & \cdots & 0 \\
     0 & p & 0 & & 0 \\
     0 & 0 & p & & 0 \\
     \vdots & \vdots &\vdots & \ddots  & \vdots \\
     0 & 0 & 0 &\ldots & p 
 \end{pmatrix} \in \mathbb{Z}^{(n+1) \times n}
 \label{eq:ll-definition}
\end{equation}

Now, let us elaborate on the connection between matrix $\Ll$ and any vector
$(\iv{\mu a_1}_p,\ldots,\iv{\mu a_n}_p)$ from the statement
of~\cref{thm:multipliers}.

\begin{observation}\label{obs:mu}
    For every integer $\mu \in \Z$ it holds that:

    \begin{displaymath}
        (\iv{\mu a_1}_p,\ldots,\iv{\mu a_n}_p) \in \Ll.
    \end{displaymath}

    Conversely, for every vector $\V{v} \in \Ll \cap
    \{-\floor{p/2},\ldots,\floor{p/2}\}^n$, there exists $\mu' \in \Z$ such
    that:
    \begin{displaymath}
        \V{v} = (\iv{\mu' a_1}_p,\ldots,\iv{\mu' a_n}_p).
    \end{displaymath}

    Moreover, $\mu'$ is unique modulo $p$ and its exact value can be determined with a constant number of
    arithmetic operations.
\end{observation}

\begin{proof}
    For the first property observe that trivially $\mu \cdot (a_1,\ldots,a_n) \in
    \Ll$. Let $\V{e}_i$ be the vector with $1$ at the $i$th coordinate and $0$
    at the remaining coordinates. Note that by definition $(p \cdot \V{e}_i) \in
    \Ll$ for every $i \in [n]$. Then we construct the desired vector by

    \begin{displaymath}
        (\iv{\mu a_1}_p,\ldots,\iv{\mu a_n}_p) = \mu \cdot
        (a_1,\ldots,a_n) -
        \sum_{i=1}^n \floor{\frac{\mu a_i}{p}} \cdot \V{e}_i \in \Ll.
    \end{displaymath}

    For the converse property, by definition any vector $\V{v} \in \Ll$ is
    represented as:

    \begin{displaymath}
        \V{v} = k_0 \cdot (a_1,\ldots,a_n) + \sum_{i=1}^n (k_i p) \cdot \V{e}_i.
    \end{displaymath}

    for some integers $k_0,\ldots,k_n \in \Z$. Let us inspect the $i$th
    coordinate of $\V{v}$. For
    every $i \in [n]$ it holds that $k_0 \cdot a_i + k_i \cdot p \in
    \{-\floor{p/2},\ldots,\floor{p/2}\}$. Therefore, for every $i \in [n]$, the $i$th coordinate of
    $\V{v}$ is $\iv{k_0 a_i}_p$. We hence can set $\mu' \coloneqq
    k_0$. Note, that given vector $\V{v}$, the integer $\mu'$ can be computed
    efficiently because it is expressed as $\mu' = v_1 \cdot a_1^{-1}$, where
    $v_1$ is the first coordinate of $\V{v}$.
\end{proof}

Now, let us further elaborate on the subsequent steps. We run the LLL algorithm on $\Ll$.
The LLL algorithm returns a basis $\V{b}_1,\ldots,\V{b}_n$ of $\Ll$. We show that,
with high probability,
$\norm{\V{b}_i}_1 \le
p/2$ for every $i \in [n]$. If that occurs we are nearly finished.
By the~\cref{obs:mu} each of the basis vectors $\V{b}_i$ is equal to $(\iv{\mu_i
a_1}_p,\ldots,\iv{\mu_i a_n}_p)$ for some integer $\mu_i \in \Z$. Moreover,
because the vectors $\V{b}_1,\ldots,\V{b}_n$ form a basis of $\Ll$ it means that the matrix
formed by the vectors $\V{b}_1,\dots,\V{b}_n$ has a full-rank.

Therefore, to complete the proof of~\cref{thm:multipliers} it suffices to show
that $\norm{\V{b}_i}_1 \le p/2$ for every $i \in [n]$. By Cauchy-Schwartz
inequality it is actually sufficient to prove that $\norm{\V{b}_i}_2 \le
\frac{p}{2\sqrt{n}}$.

\subsection{Basis of Generating Family and Dual Basis}

Lattice $\Ll$ is given to us by a generating family in form of a
rectangular matrix. It is much more convenient to work with a basis,
especially when given by a square matrice. Therefore, we start by
determining a basis of $\Ll$. As we have already noticed, we can
assume without loss of generality that $a_1$ is invertible in
$\mathbb{Z}_p$. For every $i \in {2,\ldots,n}$, let
$\alpha_i = a_i a^{-1}_1$. The rows of the following matrix form a
basis of $\Ll$:

\iftrue
 \begin{displaymath}
     \Bb_0 = 
 \begin{pmatrix}
     1 & \alpha_2 & \alpha_3 & \cdots & \alpha_n \\
     0 & p & 0 & & 0 \\
     0 & 0 & p & & 0 \\
     \vdots & \vdots &\vdots & \ddots  & \vdots \\
     0 & 0 & 0 &\ldots & p 
 \end{pmatrix} \in \Z^{n \times n}
 \end{displaymath}
\else
\begin{displaymath}
    \Bb = 
\begin{pmatrix}
    1 & 0 & 0 & \cdots & 0 \\
    \alpha_2 & p & 0 & & 0 \\
    \alpha_3  & 0 & p & & 0 \\
    \vdots & \vdots &\vdots & \ddots  & \vdots \\
    \alpha_n & 0 & 0 &\ldots & p 
\end{pmatrix} \in \Z^{n \times n}
\end{displaymath}
\fi

In particular, this implies that the volume of $\Ll$ is $p^{n-1}$.
We can determine the dual lattice of $\Ll$, which in the
full-dimensional case is spanned by the rows of the transpose of the inverse matrix.
In our case, this is:
\iftrue
 \begin{displaymath}
     \Bb_0^\dagger = 
     (\Bb_0^{-1})^T = 
     \frac{1}{p} \cdot 
     \begin{pmatrix}
         p & 0 & 0 & \cdots & 0 \\
         -\alpha_2 & 1 & 0 &\cdots & 0 \\
         -\alpha_3 & 0 & 1 & & 0 \\
         \vdots & \vdots &\vdots & \ddots  & \vdots \\
         -\alpha_n & 0 & 0 &\ldots & 1 
     \end{pmatrix}
     \in \mathbb{Q}^{n \times n}
   \end{displaymath}
\else
\begin{displaymath}
    \Bb^\dagger = 
    \Rev((\Bb^{-1})^T) = 
    \frac{1}{p} \cdot 
    \Rev 
    \begin{pmatrix}
        p & -\alpha_2 & -\alpha_3 & \cdots & -\alpha_n \\
        0 & 1 & 0 &\cdots & 0 \\
        0 & 0 & 1 & & 0 \\
        \vdots & \vdots &\vdots & \ddots  & \vdots \\
        0 & 0 & 0 &\ldots & 1 
    \end{pmatrix}
    \in \mathbb{Q}^{n \times n}
    .
  \end{displaymath}
\fi
Observe that $p\Bb_0^\dagger$ generates the set of integral vectors
$(x_1,\ldots,x_n) \in \mathbb{Z}^n$ such that
$\sum_{i=2}^n x_i \cdot \alpha_i \equiv_p -x_1$. Or equivalently such
that:
\[
  \sum_{i=1}^{n}x_i a_i\equiv_p 0.
 \]

\subsection{Every vector of the LLL-reduced basis is probably short}
Now, let $\Bb$ be an LLL-reduced basis of the
lattice $\Ll$. We denote its row vectors by
$\V{b}_1,\ldots,\V{b}_n$ and their GSO by
$\V{b}_1^\ast,\ldots,\V{b}_n^\ast$.
Recall that the dual basis, given by
$\Bb^\dagger =    \Rev((\Bb^{-1})^T)$.


\noindent In this section, we prove the following probabilistic
property on a reduced basis:

\begin{lemma}
  \label{lem:main-bound}
  With probability $\ge 1-2^{-\Omega(n \log{n})}$, it holds that:
    \[
        \norm{\V{b}_k^\ast} \le (\CLLL)^{\frac{n-1}{2}} \Vol(\Ll)^{1/n}
     \quad \text {for every} \  k \in \{1,\ldots,n\}.
    \]
\end{lemma}

\cref{lem:main-bound} can be proved as a consequence of the following
inequalities, which hold with overwhelming probability when the integers
$a_1,\ldots,a_n$ are drawn uniformly at random.

\begin{claim}
    \label{event1}
    With probability $\ge 1-2^{-\Omega(n \log{n})}$ it holds that:
    \begin{displaymath}
        \norm{\V{b}_n^\ast} \le \Vol(\Ll)^{1/n} \le \norm{\V{b}_1^\ast}
        .
    \end{displaymath}
\end{claim}
\begin{proof}[Proof of~\cref{lem:main-bound} assuming~\cref{event1}]
    Let us assume that both inequalities in~\cref{event1} hold. If $k \ge
    \frac{n+1}{2}$ then by repeated application of Inequality~\eqref{ineq:LLL3} we have that
    $\norm{\V{b}_k^\ast} \le (\CLLL)^{n-k} \norm{\V{b}_n^\ast} \le
    (\CLLL)^{\frac{n-1}{2}} \Vol(\Ll)^{1/n}$, which is the desired result. Hence we need to focus on the case
    where $k < \frac{n+1}{2}$.
    Observe that by repeated application of Inequality~\eqref{ineq:LLL3} and~\cref{event1} we have
    \begin{align*}
        \norm{\V{b}_1^\ast} \cdots \norm{\V{b}_{k-1}^\ast} & \ge
        \prod_{i=1}^{k-1} \frac{\norm{\V{b}_1^\ast}}{(\CLLL)^{i-1}} =
        (\CLLL)^{-\frac{(k-1)(k-2)}{2}}\cdot \norm{\V{b}_1^\ast}^{k-1} \ge
        (\CLLL)^{-\frac{(k-1)(k-2)}{2}}\cdot \Vol(\Ll)^\frac{k-1}{n} .
    \intertext{Similarly, we have}
        \norm{\V{b}_k^\ast} \cdots \norm{\V{b}_n^\ast} & \ge \prod_{i=0}^{n-k}
        \frac{\norm{\V{b}_k^\ast}}{(\CLLL)^i} = (\CLLL)^{-\frac{(n-k)(n-k+1)}{2}}\cdot
        \norm{\V{b}_k^\ast}^{n-k+1}.
    \intertext{Recall that the volume of the lattice $\Ll$ is $\Vol(\Ll) =
    \norm{\V{b}_1^\ast} \cdots \norm{\V{b}_n^\ast}$.
    Hence by multiplying the above inequalities we have}
        \Vol(\Ll) & \ge \Vol(\Ll)^{\frac{k-1}{n}} \cdot
        \norm{\V{b}_k^\ast}^{n-k+1} \cdot (\CLLL)^{-\frac{(n-k)(n-k+1)}{2} -\frac{(k-1)(k-2)}{2}}.
    \intertext{Because $k < \frac{n+1}{2}$ we have that $(n-k)(n-k+1) + (k-1)(k-2) \le
    (n-1)(n-k+1)$. Hence}
        \Vol(\Ll) & \ge \Vol(\Ll)^{\frac{k-1}{n}} \cdot
        \norm{\V{b}_k^\ast}^{n-k+1} \cdot (\CLLL)^{-\frac{(n-1)(n-k+1)}{2}}.
    \end{align*}
    By rearranging the terms we have
    \begin{displaymath}
        \Vol(\Ll)^{\frac{n-k+1}{n}} \cdot (\CLLL)^{\frac{(n-1)(n-k+1)}{2}} \ge \norm{\V{b}_k^\ast}^{n-k+1} .
    \end{displaymath}
    After taking $(n-k+1)$th root, this yields the desired inequality.
\end{proof}

It remains to prove inequalities in~\cref{event1}. We split this proof
into~\cref{prop:event-1} and~\cref{prop:event-2}. First, we focus on the right
side of the inequality.

\begin{proposition}\label{prop:event-1}
    $$\prob{\norm{\V{b}_1^\ast}<\Vol(\Ll)^{1/n} } < 2^{-\Omega(n\log{n})}.$$
\end{proposition}

\begin{proof}
    Recall that $\Vol(\Ll) = p^{n-1}$ and $\V{b}_1 = \V{b}_1^\ast$. Let us bound the probability of
    $\norm{\V{b}_1} < p^{\frac{n-1}{n}}$. Note that in this case, \cref{obs:mu} asserts that there exists
    $\mu \in \Z$ such that

    \begin{displaymath}
        \norm{(\iv{\mu a_1}_p, \ldots, \iv{\mu a_n}_p)} < p^{\frac{n-1}{n}}.
    \end{displaymath}

    Moreover, integer $\mu$ is determined exactly as $\mu \coloneqq (\V{b}_1)_1
    \cdot a_1^{-1}$. Hence, for a fixed $\V{y} \in
    \{-\floor{p/2},\ldots,\floor{p/2}\}^n$ the
    numbers $a_2,\ldots,a_n$ via $\iv{\cdot}_p$ operations are determined. This
    means that for a given value $a_1$, the conditional probability for a fixed $\V{y} \in
    \{-\floor{p/2},\ldots,\floor{p/2}\}^n$ is:

    \begin{displaymath}
        \prob{\V{y} = (\iv{\mu a_1}_p, \ldots, \iv{\mu a_n}_p)} = \frac{1}{p^{n-1}}.
    \end{displaymath}

    which does not depend on $a_1$. Thus the unconditional probability is
    $1/p^{n-1}$ as well.  On the other hand, the number of vectors $\V{y}$ such
    that $\norm{\V{y}} < p^{\frac{n-1}{n}}$ is $\le
	\Vol(\Ball_n(p^{\frac{n-1}{n}}+\sqrt{n}))$. Therefore by the union bound we
    conclude:
    \begin{displaymath}
        \prob{\norm{\V{b}_1} \le p^{\frac{n-1}{n}}} <
		\frac{\Vol(\Ball_n(p^{\frac{n-1}{n}}+\sqrt{n}))}{p^{n-1}} < 2^{-\Omega(n \log{n})}.
        \qedhere
    \end{displaymath}
\end{proof}

For the other part of the inequality in~\cref{event1}, we need to prove that:
\begin{proposition}\label{prop:event-2}
    $$\prob{\norm{\V{b}_n^\ast} > \Vol(\Ll)^{1/n}} < 2^{-\Omega(n \log{n})}.$$
\end{proposition}
\begin{proof}
    Again recall that $\Vol(\Ll) = p^{n-1}$ and our goal is to bound the
    probability that $\norm{\V{b}_n^\ast} > p^{1-1/n}$. Now, we inspect the dual
    lattice $\Ll^\dagger$ and the dual basis $\Bb^\dagger$ of the
    reduced basis. Recall that $(\V{b}_1^\dagger)^\ast = \V{b}_1^\dagger$ and 
    $\norm{\V{b}_n^\ast} = \norms{\V{b}_1^\dagger}^{-1}$. Therefore, we aim to bound the
    probability $\norms{\V{b}_1^\dagger} < \frac{p^{1/n}}{p}$. 

    Because $\V{b}_1^\dagger \in \Ll^\dagger$, there exists $\V{s} =
    (s_1,\ldots,s_n) \in \Z^n$ such
    that:

    \begin{displaymath}
        \V{b}_1^\dagger = \frac{1}{p} \cdot \left(\left(p \cdot s_1 - \sum_{i=2}^n
            \alpha_i \cdot  s_i \right),
        s_2,s_3,\ldots,s_n\right).
    \end{displaymath}

    First, remark that when
    the length of $\V{b}_1^\dagger$ is bounded by $\frac{p^{1/n}}{p}$ it holds that:
    \begin{displaymath}
        \norm{(s_2,\ldots,s_n)} < p^{1/n}.
    \end{displaymath}

    Observe that the number of $(s_2,\ldots,s_n) \in \Z^{n-1}$ of
	length $< p^{1/n}$ is at most $\Vol(\Ball_{n-1}(p^{1/n}+\sqrt{n}))$. Hence, from now on, we
    fix coordinates $s_2,\ldots,s_n$ and examine the probability that the
    first coordinate of $\V{b}_1^\dagger$ is bounded. In particular, when
    $\norms{\V{b}_1^\dagger} < \frac{p^{1/n}}{p}$ it holds that:

    \begin{displaymath}
        \left| p\, s_1 - \sum_{i=2}^n \alpha_i s_i \right| < p^{1/n}
    \end{displaymath}

    Note that, if the values $s_2,\ldots,s_n$ are fixed then $s_1$ is
    determined. Namely, $s_1$ is selected in such a way that:
    \begin{equation}\label{eq:unlikely-bound}
        \sum_{i=2}^n \alpha_i s_i \equiv k \Mod{p} \text{ for some } k \in [p^{1/n}]
    \end{equation}
    Recall that the numbers $a_1,\ldots,a_n$ are selected uniformly at random
    from a range greater than $[p]$. Hence, for a fixed $\V{s}$ the
    probability that Equation~\ref{eq:unlikely-bound} holds is
    $\frac{p^{1/n}}{p}$.  Therefore, by union bound we have that
    \begin{displaymath}
        \prob{\norms{\V{b}_1^\dagger} < \frac{p^{1/n}}{p}} \le \frac{p^{1/n}}{p}
		\cdot \Vol(\Ball_{n-1}(p^{1/n}+\sqrt{n})) < 2^{-\Omega(n \log{n})}.
        \qedhere
    \end{displaymath}
\end{proof}

Applying the union bounds on the two propositions gives an upper bound
on the negation of the event considered in \cref{event1}. This concludes the
proof of~\cref{event1}.

\subsection{Proof of~\cref{thm:multipliers}}
By~\cref{lem:main-bound}, with probability $\ge 1-2^{-\Omega(n\log{n})}$, for
every $i \in [n]$ we have:
\begin{displaymath}
    \norm{\V{b}_i^\ast} \le (\CLLL)^{\frac{n-1}{2}} p^{\frac{n-1}{n}}.
\end{displaymath}
We set $p = \Theta((\CLLL)^{n^2/2} \cdot 2^{2 n
\log_2{n}})$. This means that for every $i \in [n]$ it holds that:
\begin{displaymath}
    \norm{\V{b}_i^\ast} \le \frac{p}{n^2}.
\end{displaymath}
The vectors $\V{b}_1^\ast,\ldots,\V{b}_n^\ast$ are the GSO basis of
$\V{b}_1,\ldots,\V{b}_n \in \Z^n$ and do not necessarily have integral
components. We want to bound the lengths of the $\V{b}_1,\ldots,\V{b}_n$ in the LLL-reduced basis. 
By definition of GSO basis, we have
\begin{displaymath}
    \V{b}_i = \V{b}^\ast_i + \sum_{j=1}^{i-1} \mu_{i,j} \V{b}_j^\ast.
\end{displaymath}

The size reduction condition of LLL guarantees that $|\mu_{i,j}| \le \mu < 1$.
Therefore, for every $i \in [n]$ it holds that:
\begin{displaymath}
    \norm{\V{b}_i} \le n \cdot \max_{k \in [n]} \norm{\V{b}_k^\ast} <\frac{p}{n}
    .
\end{displaymath}
By the Cauchy-Schwartz inequality, the fact that $\norm{\V{b}_i} < p/n$
implies that $\norm{\V{b}_i}_1 < p/\sqrt{n}$. In particular, it means that $\V{b}_i \in
\{-\floor{p/2},\ldots,\floor{p/2}\}^n$. Therefore, by~\cref{obs:mu} we can find
integers $\mu_1,\ldots,\mu_n$ such that:
\begin{displaymath}
    \V{b}_i = (\iv{\mu_i a_1}_p,\ldots,\iv{\mu_i a_n}_p).
\end{displaymath}
for every $i \in [n]$. Note that these vectors are linearly independent because
the vectors $\V{b}_i$ form a basis.

For the running time, observe that we need a single call to the LLL algorithm to get the basis
$\V{b}_1,\ldots,\V{b}_n$ and a linear number of arithmetic operations to retrieve
the coefficients $\mu_1,\ldots,\mu_n$.
This concludes the proof of~\cref{thm:multipliers}.\qed

%% file: chapters/block-reduction.tex
\section{Extension to the block reduction}\label{sec:block}

In this Section, we prove~\cref{thm:block}. Note that the proof
of~\cref{thm:main-theorem} and the proof of~\cref{thm:multipliers} only relied on the
specific lattice reduction via~\cref{lem:main-bound}. Hence, in
order to establish~\cref{thm:block}, it is enough to
generalize~\cref{lem:main-bound} to the case of block basis reduction and prove the
following statement.

\begin{lemma}\label{lem:main-bound-BKZ}
  Let $\V{b}_1,\ldots,\V{b}_n$ be the reduced basis obtained by running the
  block lattice reduction algorithm of~\cite{gama2008finding} with parameter
  $\CGamma$ on the lattice $\Ll$ defined in~\eqref{eq:ll-definition}.
  With probability $\ge 1-2^{-\Omega(n \log{n})}$, it holds that:
    \[
        \norm{\V{b}_k^\ast} \le C_{\CGamma} \cdot \CGamma^{\frac{n-1}{2}} \cdot \Vol(\Ll)^{1/n}
     \quad \text {for every} \  k \in \{1,\ldots,n\},
    \]
    where $C_{\CGamma}$ is a constant that only depends on $\CGamma$.
\end{lemma}

Before we prove~\cref{lem:main-bound-BKZ}, let us elaborate on the block basis
reduction from~\cite{gama2008finding}. Here, we follow the description from the
textbook~\cite[Chapter 2, Sliding Algorithm]{nguyen2010lll}. The main parameter in the block-reduction is the
block size $w > 2$. Gama and Nguyen~\cite{gama2008finding}, actually
parameterize their reduction with respect to the block-size. For any integer $k$
let $\gamma_k$ be the Hermite constant (see~Chapter 2~\cite{nguyen2010lll}).
The blocksize $w$ of the lattice reduction of Gama and Nguyen is the smallest
integer such that $\CGamma \ge (\gamma_w)^{1/(w-1)}$. Note that $\gamma_n =
\Theta(n)$, so $w$ is a properly defined constant that depends on the choice of
$\CGamma$.

The block-reduction of~\cite{gama2008finding} has two important properties.
First, it returns a basis that is \emph{block-Mordell-reduced}. The only
property about block-Mordell-reduced basis we need is the following
inequality:

\begin{claim}[Primal-Dual inequality, Chapter 2, Lemma 11, Equality (2.48) in \cite{nguyen2010lll}]
    Let $\V{b}_1,\ldots,\V{b}_n$ be a block-Mordell-reduced basis of the lattice
    $L$ with blocksize $w \ge 2$, then:
    \begin{align}\label{ineq:block-ineq}
        \frac{\norm{\V{b}_j^\ast}}{\norm{\V{b}_{j+w}^\ast}} \le (\gamma_w)^{w/(w-1)} \le \CGamma^w
    \end{align}
    where $j \in \{1,\ldots,n-w\}$ such that $j \equiv_w 1$.
\end{claim}

The second property of the algorithm from~\cite{gama2008finding} is that its
output basis is also LLL-reduced (see~\cite[Chapter 2, Sliding Algorithm,
Algorithm 6]{nguyen2010lll}). As a consequence, we can also use inequality~\eqref{ineq:LLL3}.

\begin{claim}
    For all $1 \le j < i \le n$, it holds that
    \begin{align}\label{ineq:inside2}
        \norm{\V{b}_i^\ast} \ge (\CLLL)^{-2w} \cdot \CGamma^{j-i} \cdot \norm{\V{b}_j^\ast}
        .
    \end{align}
\end{claim}
\begin{proof}

    Note that blocks in block-reduced algorithm overlap on the indices $\equiv_w
    1$~\cite{gama2008finding}. Hence, for the sake of clarity for $i \in [n]$
    let $q(i) \in \N$ to be the index of the block that $i$ is contained, i.e.,
    $i \in [q(i) w+1,(q(i)+1)w]$.  Let $\ell(i) = q(i)w + 1$ be the index of the
    first element in the $i$th block. Similarly let $r(i) = \ell(i)$ when $i
    \equiv_w 1$ and $r(i) = \min\{n,\ell(i)+1\}$ otherwise, be the index of the
    last element of the block.

    Because $\V{b}_1,\ldots,\V{b}_n$ is also LLL reduced we know that that by
    repeated application of inequality~\eqref{ineq:LLL3}, for every $k \in [n]$
    it holds:
    \begin{align}\label{eq:block-l-k}
        (\CLLL)^{-w} \norms{\V{b}_{\ell(k)}^\ast} \le \norm{\V{b}_k^\ast} \le (\CLLL)^{w} \norms{\V{b}_{r(k)}^\ast}.
    \end{align}
    Next, we use~\cref{ineq:block-ineq}, that for $j < i$ says:
    \begin{align}\label{eq:blocks-cgamma}
        \norms{\V{b}_{\ell(i)}^\ast} \ge  \CGamma^{(\ell(i) - r(j)) w} \cdot \norms{\V{b}_{r(j)}^\ast}.
    \end{align}
    Hence, by combining above we get:
    \begin{align*}
        \norms{\V{b}_i^\ast} & \ge (\CLLL)^{-w} \cdot
        \norms{\V{b}^\ast_{\ell(i)}} &
        \color{gray}{\text{(by~\eqref{eq:block-l-k})}}\\
        & \ge (\CLLL)^{-w} \cdot \CGamma^{(\ell(i) - r(j)) w} \cdot
        \norms{\V{b}_{r(j)}^\ast}& 
        \color{gray}{\text{(by~\eqref{eq:blocks-cgamma})}}\\
        & \ge (\CLLL)^{-2w} \cdot \CGamma^{i-j} \cdot \norms{\V{b}^\ast_j},&
        \color{gray}{\text{(by~\eqref{eq:block-l-k})}}\\
    \end{align*}
    where the last inequality follows because $\ell(i)w - r(j)w \ge i-j$.
    %
\end{proof}

\begin{proof}[Proof of~\cref{lem:main-bound-BKZ}]
    The probabilistic event in \cref{event1} is a function of the lattice $\Ll$
    itself and not of any specific basis of it. Thus, it remains valid for the basis
    produced by the block reduction algorithm.  Hence, with $\ge 1-2^{-\Omega(n
    \log{n})}$ we have that $\norm{\V{b}_n^\ast} \le \Vol(\Ll)^{1/n} \le
    \norm{\V{b}_1^\ast}$. First, consider the case $k \ge \frac{n+1}{2}$.
    \begin{displaymath}
        \norm{\V{b}_k^\ast} \le (\CLLL)^{2w} \cdot \CGamma^{n-k} \norm{\V{b}_n^\ast}.
    \end{displaymath}
    This concludes the proof for the case $k \ge \frac{n+1}{2}$ as $n-k \le (n-1)/2$.
    It remains to consider the case $k < \frac{n+1}{2}$. By repeated
    application of~\eqref{ineq:inside2} we have:
    \begin{align*}
        \norm{\V{b}_1^\ast} \cdots \norm{\V{b}_{k-1}^\ast} & \ge
        \prod_{i=1}^{k-1} (\CLLL)^{-2w} \cdot \CGamma^{-(i-1)} \norm{\V{b}_1^\ast}\\
        & \ge (\CLLL)^{-2w(k-1)} \cdot \CGamma^{-(k-2)(k-1)/2} \cdot
        \norm{\V{b}_1^\ast}^{k-1}.
    \end{align*}
    Similarly:
    \begin{align*}
        \norm{\V{b}_k^\ast} \cdots \norm{\V{b}_{n}^\ast} & \ge 
        \prod_{i=k}^{n} (\CLLL)^{-2w} \cdot \CGamma^{-(i-k)} \norm{\V{b}_k^\ast}\\
        & \ge (\CLLL)^{-2w(n-k)} \cdot \CGamma^{-(n-k)(n-k+1)/2} \cdot
        \norm{\V{b}_k^\ast}^{n-k}.
    \end{align*}
    Hence, by multiplying these two inequalities and by $\prod_{i=1}^n \norm{\V{b}_i^\ast} = \Vol(\Ll)$ we conclude that:
    \begin{align*}
        \Vol(\Ll) \ge (\CLLL)^{-2w(n-1)} \cdot \CGamma^{-(n-k)(n-k+1)/2 -
        (k-2)(k-1)/2} \cdot \norm{\V{b}_k^\ast}^{n-k+1} \cdot \norm{\V{b}_1^\ast}^{k-1}.
    \end{align*}
    Because $k < \frac{n+1}{2}$ we have that $(n-k)(n-k+1) + (k-1)(k-2) \le
    (n-1)(n-k+1)$. Hence,
    \begin{align*}
        \Vol(\Ll) \ge (\CLLL)^{-2w(n-1)} \cdot \CGamma^{-(n-1)(n-k+1)/2}
        \cdot \norm{\V{b}_k^\ast}^{n-k+1} \norm{\V{b}_1^\ast}^{k-1}.
    \end{align*}
    Next, we use~\cref{event1} with $\norm{\V{b}_1^\ast} \ge \Vol(\Ll)^{1/n}$ to have
    \begin{align*}
        \Vol(\Ll)^{\frac{n-k+1}{n}} \ge (\CLLL)^{-2w(n-1)} \cdot \CGamma^{-((n-1)(n-k+1)/2)}
        \cdot \norm{\V{b}_k^\ast}^{n-k+1}.
    \end{align*}
    After taking $(n-k+1)$th root we conclude that:
    \begin{align*}
        (\CLLL)^{2w(n-1)/(n-k+1)} \cdot \CGamma^{(n-1)/2} \cdot \Vol(\Ll) \ge \norm{\V{b}_k^\ast}.
    \end{align*}
    Finally, observe that as $k < (n+1)/2$ this means that
    $(\CLLL)^{2w(n-1)/(n-k+1)} \le (\CLLL)^{4w} = C_{\CGamma}$ is a constant
    that depends only on $\CGamma$, since $\CLLL$ is a constant and $w$ depends
    only on $\CGamma$. Hence:
    \begin{align*}
        C_{\CGamma} \cdot \CGamma^{(n-1)/2} \cdot \Vol(\Ll)^{1/n} \ge
        \norm{\V{b}_k^\ast}. & \qedhere
    \end{align*}
\end{proof}

%% file: chapters/conclusion.tex
\section{Conclusion}

Subset Sum is one of the most fundamental problems in theoretical computer
science. To this day, there are three methods to solve it:
\begin{itemize}
        \setlength\itemsep{0.2em}
    \item \textbf{Dense regime}: Dynamic programming in $\Os(T)$, where $T$ is the target value,
    \item \textbf{Threshold regime}: Meet-in-the-middle algorithm that runs in
        $\Os(2^{n/2})$ time,
    \item \textbf{Sparse regime}: Lattice-reduction, for numbers of size
        $\LORange = \Omega(\gamma^{n^2})$.
\end{itemize}
After almost 40 years, we are aware only of improvements to the exponent in the
Subset Sum problem. The exact $\Os(2^{n/2})$ algorithm was enhanced to
$\Oh(2^{0.291n})$ in the average case setting~\cite{generic-knapsack2}.  In the
dense regime, the potential improvements are even smaller, as we know that no
$T^{1-\epsilon} \cdot \poly(n)$ algorithm is possible for any constant $\eps >
0$ assuming SETH~\cite{abboud2022seth}. Hence, researchers have focused on optimizing
polynomial factors in $n$ (see
e.g.,~\cite{bringmann2017near,chen2024faster,jin20230,bringmann2023knapsack}).

In this paper, we improved the Lagarias-Odlyzko algorithm for average-case
Subset Sum and provided a method that works for almost all instances with
numbers of size $\Omega(\sqrt{\LORange})$. In terms of the constant in the
exponent, this represents one of the most substantial improvements to the regime
of solvable instances since the 1980s.

We would like to conclude with potential improvements and applications. For
example, Schnorr~\cite{schnorr} presented a basis reduction that guarantees a
$2^{\Oh(n \log \log n/\log n)}$-approximation to CVP and theoretically improves
LLL. Our method incurs essentially no additional and can just be plugged in even
with these more computationally expensive lattice-reductions to improve the
density range (see~\cref{sec:block} where we combine it with textbook variant of
Schnorr~\cite{schnorr} lattice-reduction). However, to achieve even higher improvements,
substantially new ideas would need to be presented as improving lattice
reduction of Schnorr~\cite{schnorr} is a major open question (see,
e.g.,~\cite{DBLP:conf/eurocrypt/AggarwalLS21} for recent progress on CVP). We
believe that similarly to the Lagarias-Odlyzko algorithm, our methods have
potential to be applicable outside of the realm of Knapsack problems (see for
example,~\cite{zadik2022lattice,lyubashevsky2005parity,lyubashevsky2005parity}).

%% file: main.bbl
\begin{thebibliography}{10}

\bibitem{abboud2022seth}
A.~Abboud, K.~Bringmann, D.~Hermelin, and D.~Shabtay.
\newblock {SETH-based lower bounds for subset sum and bicriteria path}.
\newblock {\em ACM Transactions on Algorithms (TALG)}, 18(1):1--22, 2022.

\bibitem{DBLP:conf/eurocrypt/AggarwalLS21}
D.~Aggarwal, Z.~Li, and N.~Stephens{-}Davidowitz.
\newblock {A {$2^{n/2}$}-Time Algorithm for {$\sqrt{n}$}-SVP and
  {$\sqrt{n}$}-Hermite SVP, and an Improved Time-Approximation Tradeoff for
  {(H)SVP}}.
\newblock In {\em Advances in Cryptology - {EUROCRYPT} 2021 - 40th Annual
  International Conference on the Theory and Applications of Cryptographic
  Techniques}, volume 12696 of {\em Lecture Notes in Computer Science}, pages
  467--497. Springer, 2021.

\bibitem{DBLP:conf/esa/AllcockHJKS22}
J.~Allcock, Y.~Hamoudi, A.~Joux, F.~Klingelh{\"{o}}fer, and M.~Santha.
\newblock Classical and quantum algorithms for variants of subset-sum via
  dynamic programming.
\newblock In {\em 30th Annual European Symposium on Algorithms, {ESA} 2022},
  volume 244 of {\em LIPIcs}, pages 6:1--6:18. Schloss Dagstuhl -
  Leibniz-Zentrum f{\"{u}}r Informatik, 2022.

\bibitem{subsetsum-tradeoff}
P.~Austrin, P.~Kaski, M.~Koivisto, and J.~M{\"{a}}{\"{a}}tt{\"{a}}.
\newblock {S}pace-{T}ime {T}radeoffs for {S}ubset {S}um: {A}n {I}mproved
  {W}orst {C}ase {A}lgorithm.
\newblock In {\em Automata, Languages, and Programming - 40th International
  Colloquium, {ICALP} 2013}, pages 45--56, 2013.

\bibitem{stacs2015}
P.~Austrin, P.~Kaski, M.~Koivisto, and J.~Nederlof.
\newblock Subset {S}um in the {A}bsence of {C}oncentration.
\newblock In {\em 32nd International Symposium on Theoretical Aspects of
  Computer Science, {STACS} 2015}, pages 48--61, 2015.

\bibitem{stacs2016}
P.~Austrin, P.~Kaski, M.~Koivisto, and J.~Nederlof.
\newblock {D}ense {S}ubset {S}um {M}ay {B}e the {H}ardest.
\newblock In {\em 33rd Symposium on Theoretical Aspects of Computer Science,
  {STACS} 2016}, pages 13:1--13:14, 2016.

\bibitem{polyspace-stoc2017}
N.~Bansal, S.~Garg, J.~Nederlof, and N.~Vyas.
\newblock Faster {S}pace-{E}fficient {A}lgorithms for {S}ubset {S}um, k-{S}um,
  and {R}elated {P}roblems.
\newblock {\em {SIAM} J. Comput.}, 47(5):1755--1777, 2018.

\bibitem{generic-knapsack2}
A.~Becker, J.-S. Coron, and A.~Joux.
\newblock {Improved generic algorithms for hard knapsacks}.
\newblock In {\em Annual International Conference on the Theory and
  Applications of Cryptographic Techniques}, pages 364--385. Springer, 2011.

\bibitem{bellman}
R.~Bellman.
\newblock {\em Dynamic Programming}.
\newblock Princeton University Press, Princeton, NJ, USA, 1957.

\bibitem{crypto-book-2}
M.~R. Bremner.
\newblock {\em Lattice basis reduction: an introduction to the LLL algorithm
  and its applications}.
\newblock CRC Press, 2011.

\bibitem{brickell1984solving}
E.~F. Brickell.
\newblock Solving low density knapsacks.
\newblock In {\em Advances in Cryptology: Proceedings of Crypto 83}, pages
  25--37. Springer, 1984.

\bibitem{bringmann2017near}
K.~Bringmann.
\newblock {A Near-Linear Pseudopolynomial Time Algorithm for Subset Sum}.
\newblock In {\em Proceedings of the Twenty-Eighth Annual ACM-SIAM Symposium on
  Discrete Algorithms}, pages 1073--1084. SIAM, 2017.

\bibitem{bringmann2023knapsack}
K.~Bringmann.
\newblock Knapsack with small items in near-quadratic time.
\newblock {\em Accepted to STOC 2024}, 2024.

\bibitem{chen2024faster}
L.~Chen, J.~Lian, Y.~Mao, and G.~Zhang.
\newblock Faster algorithms for bounded knapsack and bounded subset sum via
  fine-grained proximity results.
\newblock In {\em Proceedings of the 2024 Annual ACM-SIAM Symposium on Discrete
  Algorithms (SODA)}, pages 4828--4848. SIAM, 2024.

\bibitem{average-case-balancing}
X.~Chen, Y.~Jin, T.~Randolph, and R.~A. Servedio.
\newblock Average-case subset balancing problems.
\newblock In {\em Proceedings of the 2022 Annual ACM-SIAM Symposium on Discrete
  Algorithms (SODA)}, pages 743--778.

\bibitem{DBLP:conf/eurocrypt/Coppersmith96}
D.~Coppersmith.
\newblock Finding a small root of a univariate modular equation.
\newblock In {\em Advances in Cryptology - {EUROCRYPT} '96, International
  Conference on the Theory and Application of Cryptographic Techniques}, volume
  1070 of {\em Lecture Notes in Computer Science}, pages 155--165. Springer,
  1996.

\bibitem{coster1992improved}
M.~J. Coster, A.~Joux, B.~A. LaMacchia, A.~M. Odlyzko, C.-P. Schnorr, and
  J.~Stern.
\newblock Improved low-density subset sum algorithms.
\newblock {\em Computational complexity}, 2:111--128, 1992.

\bibitem{dinur}
I.~Dinur, O.~Dunkelman, N.~Keller, and A.~Shamir.
\newblock {E}fficient {D}issection of {C}omposite {P}roblems, with
  {A}pplications to {C}ryptanalysis, {K}napsacks, and {C}ombinatorial {S}earch
  {P}roblems.
\newblock In {\em Advances in Cryptology - {CRYPTO} 2012 - 32nd Annual
  Cryptology Conference. Proceedings}, 2012.

\bibitem{Frieze}
A.~M. Frieze.
\newblock {On the Lagarias-Odlyzko Algorithm for the Subset Sum Problem}.
\newblock {\em {SIAM} J. Comput.}, 15(2):536--539, 1986.

\bibitem{furst1989succinct}
M.~L. Furst and R.~Kannan.
\newblock Succinct certificates for almost all subset sum problems.
\newblock {\em SIAM Journal on Computing}, 18(3):550--558, 1989.

\bibitem{crypto-book-1}
S.~D. Galbraith.
\newblock {\em Mathematics of public key cryptography}.
\newblock Cambridge University Press, 2012.

\bibitem{gama2008finding}
N.~Gama and P.~Q. Nguyen.
\newblock {Finding short lattice vectors within Mordell's inequality}.
\newblock In {\em Proceedings of the fortieth annual ACM symposium on Theory of
  computing}, pages 207--216, 2008.

\bibitem{helm2018subset}
A.~Helm and A.~May.
\newblock Subset sum quantumly in {$1.17^n$}.
\newblock In {\em 13th Conference on the Theory of Quantum Computation,
  Communication and Cryptography (TQC 2018)}. Schloss Dagstuhl-Leibniz-Zentrum
  fuer Informatik, 2018.

\bibitem{horowitz-sahni}
E.~Horowitz and S.~Sahni.
\newblock {C}omputing {P}artitions with {A}pplications to the {K}napsack
  {P}roblem.
\newblock {\em J. {ACM}}, 21(2):277--292, 1974.

\bibitem{DBLP:conf/ima/Howgrave-Graham97}
N.~Howgrave{-}Graham.
\newblock Finding small roots of univariate modular equations revisited.
\newblock In {\em Cryptography and Coding, 6th {IMA} International Conference},
  volume 1355, pages 131--142. Springer, 1997.

\bibitem{generic-knapsack1}
N.~Howgrave-Graham and A.~Joux.
\newblock {New generic algorithms for hard knapsacks}.
\newblock In {\em Annual International Conference on the Theory and
  Applications of Cryptographic Techniques}, pages 235--256. Springer, 2010.

\bibitem{jin20230}
C.~Jin.
\newblock 0-1 knapsack in nearly quadratic time.
\newblock {\em Accepted to STOC 2024}, 2024.

\bibitem{DBLP:conf/fct/JouxS91}
A.~Joux and J.~Stern.
\newblock {Improving the Critical Density of the Lagarias-Odlyzko Attack
  Against Subset Sum Problems}.
\newblock In {\em Fundamentals of Computation Theory, 8th International
  Symposium, {FCT} '91}, volume 529, pages 258--264. Springer, 1991.

\bibitem{knapsack-book}
H.~Kellerer, U.~Pferschy, and D.~Pisinger.
\newblock {\em Knapsack problems}.
\newblock Springer, 2004.

\bibitem{koiliaris2019faster}
K.~Koiliaris and C.~Xu.
\newblock {Faster pseudopolynomial time algorithms for subset sum}.
\newblock {\em ACM Transactions on Algorithms (TALG)}, 15(3):1--20, 2019.

\bibitem{LO}
J.~C. Lagarias and A.~M. Odlyzko.
\newblock Solving low-density subset sum problems.
\newblock {\em J. {ACM}}, 32(1):229--246, 1985.

\bibitem{lamacchia1991basis}
B.~A. LaMacchia.
\newblock Basis reduction algorithms and subset sum problems.
\newblock 1991.

\bibitem{lll}
A.~K. Lenstra, H.~W. Lenstra, and L.~Lov{\'a}sz.
\newblock Factoring polynomials with rational coefficients.
\newblock {\em Mathematische annalen}, 261:515--534, 1982.

\bibitem{pseudopolynomial-polyspace}
D.~Lokshtanov and J.~Nederlof.
\newblock {S}aving {S}pace by {A}lgebraization.
\newblock In {\em Proceedings of the 42nd {ACM} Symposium on Theory of
  Computing, {STOC} 2010}, pages 321--330, 2010.

\bibitem{lyubashevsky2005parity}
V.~Lyubashevsky.
\newblock The parity problem in the presence of noise, decoding random linear
  codes, and the subset sum problem.
\newblock In {\em International Workshop on Approximation Algorithms for
  Combinatorial Optimization}, pages 378--389. Springer, 2005.

\bibitem{crypto-merkle78}
R.~C. Merkle and M.~E. Hellman.
\newblock Hiding information and signatures in trapdoor knapsacks.
\newblock {\em {IEEE} Trans. Information Theory}, 24(5):525--530, 1978.

\bibitem{NederlofW21}
J.~Nederlof and K.~W\k{e}grzycki.
\newblock {Improving Schroeppel and Shamir's Algorithm for Subset Sum via
  Orthogonal Vectors}.
\newblock In {\em {STOC} '21: 53rd Annual {ACM} {SIGACT} Symposium on Theory of
  Computing}, pages 1670--1683. {ACM}, 2021.

\bibitem{nguyen2010lll}
P.~Q. Nguyen and B.~Vall{\'e}e.
\newblock {\em The LLL algorithm}.
\newblock Springer, 2010.

\bibitem{odlyzko1990rise}
A.~M. Odlyzko.
\newblock The rise and fall of knapsack cryptosystems.
\newblock {\em Cryptology and computational number theory}, 42(2), 1990.

\bibitem{icalp21}
A.~Polak, L.~Rohwedder, and K.~W\k{e}grzycki.
\newblock {Knapsack and Subset Sum with Small Items}.
\newblock In {\em 48th International Colloquium on Automata, Languages, and
  Programming, {ICALP} 2021}, volume 198 of {\em LIPIcs}, pages 106:1--106:19.
  Schloss Dagstuhl - Leibniz-Zentrum f{\"{u}}r Informatik, 2021.

\bibitem{radziszowski1988solving}
S.~Radziszowski and D.~Kreher.
\newblock Solving subset sum problems with the {$L^3$} algorithm.
\newblock {\em The Charles Babbage Research Centre: The Journal of
  Combinatorial Mathematics and Combinatorial Computing}, 3, 1988.

\bibitem{10.1007/978-3-031-38548-3_1}
K.~Ryan and N.~Heninger.
\newblock Fast practical lattice reduction through iterated compression.
\newblock In {\em Advances in Cryptology -- CRYPTO 2023}, pages 3--36, Cham,
  2023. Springer Nature Switzerland.

\bibitem{schnorr}
C.-P. Schnorr.
\newblock A hierarchy of polynomial time lattice basis reduction algorithms.
\newblock {\em Theoretical computer science}, 53(2-3):201--224, 1987.

\bibitem{schnorr1994lattice}
C.-P. Schnorr and M.~Euchner.
\newblock Lattice basis reduction: Improved practical algorithms and solving
  subset sum problems.
\newblock {\em Mathematical programming}, 66:181--199, 1994.

\bibitem{schroeppel}
R.~Schroeppel and A.~Shamir.
\newblock A {T}={$O(2^{n/2})$}, {S}={$O(2^{n/4})$} {A}lgorithm for {C}ertain
  {NP}-{C}omplete {P}roblems.
\newblock {\em {SIAM} J. Comput.}, 10(3):456--464, 1981.

\bibitem{zadik2022lattice}
I.~Zadik, M.~J. Song, A.~S. Wein, and J.~Bruna.
\newblock Lattice-based methods surpass sum-of-squares in clustering.
\newblock In {\em Conference on Learning Theory}, pages 1247--1248. PMLR, 2022.

\end{thebibliography}
